\documentclass[pra,aps,nopacs,twocolumn,twoside,superscriptaddress]{revtex4}

\usepackage{amsmath,amsfonts,amssymb,color,epsfig,graphics,graphicx,hyperref,latexsym,revsymb,theorem,verbatim}
\usepackage[utf8]{inputenc}
\usepackage[T1]{fontenc}
\usepackage{lmodern}

\newtheorem{definition}{Definition}

\newtheorem{lemma}[definition]{Lemma}

\newtheorem{theorem}[definition]{Theorem}
\newtheorem{corollary}[definition]{Corollary}

\def\squareforqed{\hbox{\rlap{$\sqcap$}$\sqcup$}}
\def\qed{\ifmmode\squareforqed\else{\unskip\nobreak\hfil
\penalty50\hskip1em\null\nobreak\hfil\squareforqed
\parfillskip=0pt\finalhyphendemerits=0\endgraf}\fi}
\def\endenv{\ifmmode\;\else{\unskip\nobreak\hfil
\penalty50\hskip1em\null\nobreak\hfil\;
\parfillskip=0pt\finalhyphendemerits=0\endgraf}\fi}
\newenvironment{proof}{\noindent \textbf{{Proof.~} }}{\qed}

\def\min{\mathop{\rm min}}
\def\max{\mathop{\rm max}}

\newcommand{\nc}{\newcommand}
\nc{\bra}[1]{\langle#1|} \nc{\ket}[1]{|#1\rangle} \nc{\proj}[1]{|
#1\rangle\!\langle #1 |} \nc{\ketbra}[2]{|#1\rangle\!\langle#2|}
\nc{\braket}[2]{\langle#1|#2\rangle} 
\nc{\norm}[1]{\lVert#1\rVert} \nc{\abs}[1]{|#1|}
\nc{\lar}{\leftarrow} \nc{\rar}{\rightarrow} \nc{\ox}{\otimes}
\nc{\op}[2]{|#1\rangle\!\langle#2|}
\nc{\ip}[2]{\langle#1|#2\rangle} \nc{\dg}{\dagger}

\nc{\cA}{{\cal A}} \nc{\cB}{{\cal B}} \nc{\cC}{{\cal C}}
\nc{\cD}{{\cal D}} \nc{\cE}{{\cal E}} \nc{\cF}{{\cal F}}
\nc{\cG}{{\cal G}} \nc{\cH}{{\cal H}} \nc{\cI}{{\cal I}}
\nc{\cJ}{{\cal J}} \nc{\cK}{{\cal K}} \nc{\cL}{{\cal L}}
\nc{\cM}{{\cal M}} \nc{\cN}{{\cal N}} \nc{\cO}{{\cal O}}
\nc{\cP}{{\cal P}} \nc{\cR}{{\cal R}} \nc{\cS}{{\cal S}}
\nc{\cT}{{\cal T}} \nc{\cX}{{\cal X}} \nc{\cZ}{{\cal Z}}

\begin{document}
\title{Optimal Entanglement Transformations Among N-qubit W-Class States}
\author{Wei Cui}
\email{cuiwei@physics.utoronto.ca}
\author{Eric Chitambar}
\email{e.chitambar@utoronto.ca}
\author{Hoi-Kwong Lo}
\email{hklo@comm.utoronto.ca}
\affiliation{Center for Quantum Information and Quantum Control (CQIQC), \\
Department of Physics and Department of Electrical \& Computer Engineering, \\
University of Toronto, Toronto, Ontario, M5S 3G4, Canada}

\begin{abstract}
We investigate the physically allowed probabilities for transforming one $N$-partite W-class state to another by means of local operations assisted with classical communication (LOCC).  Recently, Kinta\c{s} and Turgut have obtained an upper bound for the maximum probability of transforming two such states \cite{Kintas-2010a}.  Here, we provide a simple sufficient and necessary condition for when this upper bound can be satisfied and thus when optimality of state transformation can be achieved.  Our discussion involves obtaining lower bounds for the transformation of arbitrary W-class states and showing precisely when this bound saturates the bound of \cite{Kintas-2010a}.  Finally, we consider the question of transforming symmetric W-class states and find that in general, the optimal one-shot procedure for converting two symmetric states requires a non-symmetric filter by all the parties.
\end{abstract} 

\date{\today}

\maketitle

\section{Introduction}

One of the most fundamental questions in entanglement theory is whether one entangled state can be converted to another by only performing local quantum operations on each subsystem while allowing classical communication among the parties.  Protocols of this form are known as LOCC and they constitute the class of operations unable to increase the amount of entanglement in a multi-party system on average.  When one state can be transformed into another via LOCC with a non-zero probability, the two states are said to SLOCC (stochastic LOCC) related.  Much research has been devoted to studying SLOCC transformations with special attention placed on reversible SLOCC convertibility since this property provides one way of identifying states with the same type of entanglement \cite{Dur-2000a, Verstraete-2003a}.  However, SLOCC convertibility does not consider the probability of transformation success, a quantity having obvious operational importance, and beyond the bipartite pure states \cite{Vidal-1999a}, very little is known about the feasible probability rates of converting states using LOCC.

In this article, we study the optimal LOCC convertibility among $N$-qubit W-class states.  Such states are of the form $\sqrt{x_0}\ket{00\cdots 0}+\sqrt{x_1}\ket{10\cdots 0}+\sqrt{x_2}\ket{01\cdots 0}+\cdots+\sqrt{x_N}\ket{00\cdots 1}$ with the $N$-party W state $\ket{W_N}$ corresponding to $x_0=0$ and $x_i=\frac{1}{N}$ for $1\leq i\leq N$.  W-class states represent a very important family of states since they posses a high degree of robustness with respect to loss of entanglement \cite{Briegel-2001a} and non-local correlations \cite{SenDe-2003a} in the presence of noise.  Furthermore, many specific quantum cryptography and communication protocols have been designed which utilize W-type entanglement (see \cite{Wang-2009a} and references within).  Experimental setups have been proposed for the production of multiqubit W states \cite{Bastin-2009a} with the generation of $\ket{W_4}$ already realized \cite{Wieczorek-2009a}.  For the special case of three qubits, partial results concerning optimality of LOCC conversion rates have been obtained \cite{Yildiz-2010a}.

In the general multiparty setting, Kinta\c{s} and Turgut recently made significant progress in understanding LOCC transformations of W-class states \cite{Kintas-2010a}.  They prove an upper bound on the optimal probability of converting two such states, and the first part of our article derives a necessary condition for when this rate can be achieved.  We then move on to construct a general procedure for converting two W-class states which provides a lower bound on the optimal conversion probability.  The necessary condition for achieving Kinta\c{s} and Turguts' bound we obtain in the first part turns out to be sufficient when using our constructed protocol thus proving optimality.  Since much of our analysis relies on results reported in \cite{Kintas-2010a}, we will try to stay as consistent as possible with the notation established there.  In the final section, we turn to the problem of converting symmetric W-class states; i.e. those states which remain invariant under a permutation of parties.  It has been shown that two multiqubit symmetric states are related by a reversible SLOCC transformation if and only if the transformation can be accomplished by a permutation invariant SLOCC filtering operation \cite{Mathonet-2010a}.  In other words, if $\ket{\psi}$ and $\ket{\phi}$ are $N$-qubit SLOCC equivalent symmetric states: $\bigotimes_{i=1}^N A_i\ket{\psi}=\ket{\phi}$, then there exists an operator $M$ such that $M^{\otimes N}\ket{\psi}=\ket{\phi}$.  However, one question still left open is whether the same probability of transformation can be achieved in the symmetric case.  If a filter $\bigotimes_{i=1}^N A_i$ succeeds in transformation with probability $p$, does there necessarily exist a symmetric filter $M^{\otimes N}$ that transforms with the same probability?  We show that in general the answer is no and often the transformation can be achieved with a greater probability when only a single party acts non-trivially.  At the same time, we further observe the single party strategy to not be optimal in general.  These results nicely demonstrate the complexity in analyzing issues of LOCC optimality as no simple general result appears to exist, even in the symmetric multiqubit case.

For simplicity, we introduce the notation $\ket{\vec{0}}=\ket{0}^{\otimes N}$ and $\ket{\vec{i}}=\ket{0}^{\otimes i-1}\ket{1}\ket{0}^{\otimes n-i}$ so that the $N$-party W state can be expressed as $\ket{W_N}=\frac{1}{\sqrt{N}}\sum_{i=1}^N\ket{\vec{i}}$.  A state $\ket{\Psi}$ is defined as a W-class state if there exists invertible operators $A_i$ such that $\bigotimes_{i=1}^N A_i\ket{W_N}=\ket{\Psi}$.  Equivalently, a state is in the W class if it is SLOCC equivalent to the state $\ket{W_N}$ \cite{Dur-2000a}.  Each $A_i$ can be written in the form $U_i\tilde {A}_i$ with $\tilde{A}_i=\left(\begin{smallmatrix}a&b\\0&c\end{smallmatrix}\right)$ and $a,c$ being real.  Then every W-class state $\ket{\Psi}$ is of the form $\sum_{i=0}^N\sqrt{x_0}\ket{\vec{i}}$ up to the application of local unitaries; i.e. $\ket{\Psi}=\bigotimes_{i=1}^N U_i\sum_{i=0}^N\sqrt{x_i}\ket{\vec{i}}$ for unitaries $U_i$.  Furthermore, for three or more parties the coefficients $x_i$ are unique to each W-class state \cite{Kintas-2010a}.  To easily see this, observe that $\bigotimes_{i=1}^N U_i\sum_{i=0}^N\sqrt{x_i}\ket{\vec{i}}=\bigotimes_{i=1}^N V_i\sum_{i=0}^N\sqrt{x'_i}\ket{\vec{i}}$ implies $\sum_{i=0}^N\sqrt{x_i}\ket{\vec{i}}=\bigotimes_{i=1}^N W_i\sum_{i=0}^N\sqrt{x_i}\ket{\vec{i}}$ for some unitaries $W_i$.  But this means that each party's reduced state is the totally mixed state which is possible only for the bipartite state $\sqrt{\frac{1}{2}}(\ket{01}+\ket{10})$.  Thus, we can unambiguously represent every multipartite W-class state by $\ket{\vec{x}}$ where $\vec{x}=(x_1,\cdots,x_N)$ is its unique coefficient vector with $x_0=1-\sum_{i=1}^Nx_i$.

The main result presented in \cite{Kintas-2010a} is that whenever party $k$ performs a measurement on state $\ket{\vec{x}}$, for each outcome $\lambda$ occurring with probability $p_\lambda$, the components transform as 
\begin{align} 
\label{Eq:coords change}
x_j&\to s_\lambda x_j\;\;\text{for}\;j\not=k,0, &x_k&\to \{\frac{x_k}{t_\lambda},0\}
\end{align} such that $\sum_\lambda p_\lambda s_\lambda=1$ and $\sum_\lambda \frac{p_\lambda}{t_\lambda}\leq 1$.  From these relations, it follows that under any LOCC transformation with outcomes indexed by $\lambda$, the vector components are non-increasing on average:
\begin{equation}
\label{Eq:CompMono}
x_i\geq\sum_{\lambda}p_\lambda x_{i,\lambda}.
\end{equation}
Consequently, for the transformation $\ket{\vec{x}}\to\ket{\vec{y}}$, the maximum probability of success $p_{max}$ is bounded by $p_{max}\leq \min_i\{r_i\}$ where $r_i=\frac{x_i}{y_i}$.  In the remainder of the article we will assume that $r_1\leq r_2\leq\cdots\leq r_N$ as any other ordering can be accounted for by relabeling.  Thus, $p_{max}\leq r_1$ and the next two sections will prove the following result.
\begin{theorem}
For W-class states $\ket{\vec{x}}$ and $\ket{\vec{y}}$, $\ket{\vec{x}}\to\ket{\vec{y}}$ with optimal probability $r_1$ if and only if $r_2\geq r_0$.
\end{theorem}

\section{Upper Bounds}

We begin with a more detailed description of a general W-class LOCC transformation.  We can model every $m$-round LOCC protocol transforming $\ket{\vec{x}}\to\ket{\vec{y}}$ by a tree split into $m$ segments with $e^{(ij)}$ denoting the $j^{th}$ edge in the $i^{th}$ segment.  The tree begins with a single node representing the initial state $\ket{\vec{x}}$ and ends after the $m^{th}$ segment with each final node representing a different outcome state.   We will say a \textbf{branch} is any uni-directional connected path that traverses the entire length of the tree.  It is a \textbf{success branch} if its final state is $\ket{\vec{y}}$; otherwise, the branch is a \textbf{failure branch}.  An edge is called an \textbf{intermediate edge} if it contains at least one successful branch traveling through it; otherwise it is called a \textbf{failure edge}.  Let $\mathcal{I}^{(i-1,j)}$ denote the set of indices such that $k\in\mathcal{I}^{(i-1,j)}$ iff $e^{(i,k)}$ is an intermediate edge connected to edge $e^{(i-1,j)}$.  Likewise, let $\mathcal{F}^{(i-1,j)}$ denote the set of indices such that $k\in\mathcal{F}^{(i-1,j)}$ iff $e^{(i,k)}$ is a failure edge connected to edge $e^{(i-1,j)}$.  The set $\mathcal{I}^{(0,0)}$ (resp. $\mathcal{F}^{(0,0)}$) will contain the indices corresponding to the intermediate (resp. failure) edges connected to the starting node of the tree.  Finally, denote the state obtained following edge $e^{(i,j)}$ by $\ket{\vec{x}^{(i,j)}}$ with components $x^{(i,j)}_l$, and let $p_{i,j}$ be the probability of moving along edge $e^{(i,j)}$.  From the definitions, we have $\sum_{k\in\mathcal{I}^{(i-1,j)}}p_{i,k}+\sum_{k\in\mathcal{F}^{(i-1,j)}}p_{i,k}=1$ for every $i,j$.

With this formalism, we can systematically calculate the total success probability of obtaining $\ket{\vec{y}}$ from $\ket{\vec{x}}$.  A branch is successful if and only if it travels only along intermediate edges.  So the total probability is given by summing over all possible intermediate edge paths.  This value is given by
\begin{align}
\label{Eq:LOCCtotalProb}
P(\ket{\vec{x}}\to\ket{\vec{y}})&=\sum_{k_1\in\mathcal{I}^{(0,0)}}p_{1,k_1}\cdots\sum_{k_m\in\mathcal{I}^{(m-1,k_{m-1})}}p_{m,k_m}\notag\\&=\sum_{k_1\in\mathcal{I}^{(0,0)}}\cdots\sum_{k_m\in\mathcal{I}^{(m-1,k_{m-1})}}\prod_{i=1}^mp_{i,k_i}.
\end{align}
Starting from the state $\ket{\vec{x}}$ and repeatedly applying \eqref{Eq:CompMono}, we have 
\begin{widetext}
\begin{align}
\label{Eq:LOCCIneq}
x_l\geq&\sum_{k_1\in\mathcal{I}^{(0,0)}}p_{1,k_1}x_l^{(1,k_1)}+\sum_{k_1\in\mathcal{F}^{(0,0)}}p_{1,k_1}x_l^{(1,k_1)}\notag\\
\geq&\sum_{k_1\in\mathcal{I}^{(0,0)}}p_{1,k_1}\left(\sum_{k_2\in\mathcal{I}^{(1,k_1)}}p_{2,k_2}x_l^{(2,k_2)}+\sum_{k_2\in\mathcal{F}^{(1,k_1)}}p_{2,k_2}x_l^{(2,k_2)}\right)+\sum_{k_1\in\mathcal{F}^{(0,0)}}p_{1,k_1}x_l^{(1,k_1)}\notag\\
\cdots\;\;\geq&\sum_{k_1\in\mathcal{I}^{(0,0)}}\cdots\sum_{k_m\in\mathcal{I}^{(m-1,k_{m-1})}}\prod_{i=1}^mp_{i,k_i}x_l^{(m,k_m)}
+\sum_{k_1\in\mathcal{I}^{(0,0)}}\cdots\sum_{k_{m-1}\in\mathcal{I}^{(m-2,k_{m-2})}}\sum_{k_m\in\mathcal{F}^{(m-1,k_{m-1})}}\prod_{i=1}^mp_{i,k_i}x_l^{(m,k_m)}\notag\\
&+\sum_{k_1\in\mathcal{I}^{(0,0)}}\cdots\sum_{k_{m-2}\in\mathcal{I}^{(m-3,k_{m-3})}}\sum_{k_{m-1}\in\mathcal{F}^{(m-2,k_{m-2})}}\prod_{i=1}^{m-1}p_{i,k_i}x_l^{(m-1,k_{m-1})}\cdots +\sum_{k_1\in\mathcal{F}^{(0,0)}}p_{1,k_1}x_l^{(1,k_1)}.
\end{align}
\end{widetext}
This equation is quite informative since we know that for $k_m\in\mathcal{I}^{(m-1,k_{m-1})}$ we have $x_l^{(m,k_m)}=y_l$.  Then by dividing both sides of \eqref{Eq:LOCCIneq} by $y_l$ and using \eqref{Eq:LOCCtotalProb}, we have 
\begin{equation}
r_l\geq P(\ket{\vec{x}}\to\ket{\vec{y}})+ \text{``Failure Edges''}
\end{equation}
where ``Failure Edges'' refers to the non-negative quantity of all but the first term in the final inequality of \eqref{Eq:LOCCIneq}.  Physically, it is the average of the $l^{th}$ component of all failure states produced after some measurement on a success branch.  

For $P(\ket{\vec{x}}\to\ket{\vec{y}})=r_1$, this requires strict equalities in \eqref{Eq:LOCCIneq} and furthermore, the ``Failure Edge'' terms must vanish.  This latter condition means that $x_1^{(i,j)}=0$ for every failure edge $e^{(i,j)}$ connected to a success branch.  We combine these results in the following lemma.
\begin{lemma}
\label{Thm:LOCCopt}
If $\ket{\vec{x}}\to\ket{\vec{y}}$ by LOCC with probability $r_1$, then
\begin{itemize}
\item[\upshape(i)] if $e^{(i,j)}$ is an edge connected to a success branch, $x_1^{(i,j)}>0$ implies $e^{(i,j)}$ is an intermediate edge, and
\item[\upshape(ii)] $x_1^{(i-1,j)}= \sum_{k\in\mathcal{I}^{(i-1,j)}}p_{i,k}x_1^{(i,k)}$ for every $i,j$.
\end{itemize}
\end{lemma}
Combining relations \eqref{Eq:coords change} with this lemma, the following becomes apparent.
\begin{corollary}
If $\ket{\vec{x}}\to\ket{\vec{y}}$ by LOCC with probability $r_1$, there always exists at least one success branch such that the measurements along each edge satisfy $s_\lambda t_\lambda\geq 1$.
\end{corollary}
\begin{proof}
For any success branch, assume that in the first round measurement $s_\lambda t_\lambda<1$ for all $\lambda$.  Then $1=\sum_\lambda p_\lambda s_\lambda<\sum_\lambda \frac{p_\lambda}{t_\lambda}$ which is impossible.  Hence, there must be some outcome with $s_\lambda t_\lambda\geq 1$, and since the first component of the initial state is nonzero (it must be or else party one would unentangled with all other parties), the first component of the resultant state will likewise be nonzero.  Consequently, by Lemma \ref{Thm:LOCCopt}, the edge corresponding to outcome $\lambda$ is an intermediate edge.  Consider the round two measurement performed along this edge and repeat the previous argument.  This can be subsequently done for all $m$ rounds thus identifying a success branch in which all edges correspond to measurement outcomes satisfying $s_\lambda t_\lambda\geq 1$.
\end{proof}

Now for any transformation occurring with probability $r_1$, let p* denote one of the success branches described by this corollary and let its edges be $e^{(i,v_i)}$.  Along p* we can divide the protocol into three parts encoded by index sets $A$, $B$, and  $C$ where $i\in A$ if party 1 performs a measurement along $e^{(i,v_i)}$, $i\in B$ if party 2 performs a measurement along $e_{i,v_i}$, and $i\in C$ if neither parties 1 or 2 perform a measurement along $e^{(i,v_i)}$.  From \eqref{Eq:coords change}, we have the following transformations during each edge in p*:
\begin{align}
i\in& A\Rightarrow &x_1^{(i,v_i)}&=\frac{x_1^{(i-1,v_{i-1})}}{t_i}, &x_2^{(i,v_i)}&=s_i x_2^{(i-1,v_{i-1})},\notag\\ &&x_0^{(i,v_i)}&=s_i x_0^{(i-1,v_{i-1})},&&\notag\\
i\in& B\Rightarrow &x_1^{(i,v_i)}&=s_i x_1^{(i-1,v_{i-1})}, &x_2^{(i,v_i)}&=\frac{x_2^{(i-1,v_{i-1})}}{t_i},\notag\\ &&x_0^{(i,v_i)}&\geq s_i x_0^{(i-1,v_{i-1})},&&\notag\\
i\in& C\Rightarrow &x_1^{(i,v_i)}&=s_i x_1^{(i-1,v_{i-1})}, &x_2^{(i,v_i)}&=s_i x_2^{(i-1,v_{i-1})},\notag\\ 
&&x_0^{(i,v_i)}&\geq s_i x_0^{(i-1,v_{i-1})}.&&
\end{align}
This implies the following relationship between the initial and final components:
\begin{align}
y_1&=\prod_{i\in A, j\in B, k\in C}\frac{1}{t_i}s_js_k x_1,&
y_2&=\prod_{i\in A, j\in B, k\in C}s_i\frac{1}{t_j}s_k x_2,\notag\\
y_0&=\prod_{i\in A, j\in B, k\in C}s_is_js_k x_0.&
\end{align}
Substituting $y_1$ into $y_0$ yields 
\begin{equation}
\label{Eq:prod1}
y_0=\prod_{i\in A}s_it_i\frac{y_1}{x_1}x_0,
\end{equation} while dividing $y_1$ by $y_2$ gives 
\begin{equation}
\label{Eq:prod2}
\frac{y_1}{y_2}\prod_{i\in A}s_it_i=\prod_{i\in B} s_it_i\frac{x_1}{x_2}\geq \frac{x_1}{x_2}
\end{equation}
where the last inequality follows from the fact that $s_it_i\geq 1$ along every edge in p*.  Then substituting \eqref{Eq:prod1} into \eqref{Eq:prod2} gives the bound
\begin{equation}
\label{Eq:r2bound}
r_2\geq r_0.
\end{equation}

\section{Lower Bounds}

In this section, we construct a specific protocol to obtain a lower bound for the maximum probability of transforming two W-class states.  In the protocol, there will be a single success branch with edges $e^{(i)}$ and states $\ket{\vec{x}^{(i)}}$ whose $k^{th}$ component is $x_k^{(i)}$.  Only two types of measurements will be performed for each acting party $k$: TYPE 1 (T1) which has an outcome $\lambda$ such that $t_\lambda=p_\lambda$ and $x_{0,\lambda}=s_\lambda x_0$, and TYPE 2 (T2) in which $s_\lambda=\frac{1}{p_\lambda}$ for some outcome.  From \cite{Kintas-2010a}, T1 and T2 measurements can always be performed on state $\ket{\vec{x}}$ for any choice of $p_\lambda$ and $s_\lambda$ so long as $p_\lambda s_\lambda\leq 1$ and $\frac{p_\lambda}{t_\lambda}\leq 1$.  Consequently, whenever $x_k>y_k$, a T2 measurement can be performed by party $k$ with $s_\lambda=p_\lambda=1$ and $t_\lambda=r_k$.  In this case, the coordinates of parties 1 through $n$ do not change on average, and so by normalization neither does the $0^{th}$ coordinate.  By explicitly solving for the scale factor $t_\lambda$ on party $k$, we have that in a T2 measurement by party $k$ on state $\ket{\vec{x}^{(i-1)}}$ the coordinates change as:
\begin{align}
\label{Eq:coords change MOSM/PSM}
x^{(i)}_j&=\frac{x_j^{(i-1)}}{p_i}\;\;\text{for}\;j\not=k, &x^{(i)}_k&=1-\frac{1-x_k^{(i-1)}}{p_{i}}.
\end{align}

For the transformation of $\ket{\vec{x}}$ to $\ket{\vec{y}}$ it is sufficient to reach some round $i$ in which $r^{(i)}_k\geq 1$ for all $k\geq 1$.  When $r^{(i)}_k>1$, as previously noted, the $k^{th}$ party can deterministically transform the state such that $r^{(i+1)}_k=1$ and all nonzero components are unchanged.  The basic idea of the protocol described here is to systematically raise each $k^{th}$ component closer to $y_k$ one at a time in a ``piggy-back'' fashion where $r_1^{(1)}$ is first increased and made equal to $r_2^{(2)}$, then both of them are increased and made equal to $r_3^{(3)},\cdots,$ etc.  Eventually, each $k^{th}$ component will be raised to $y^k$ or possibly greater.  The next simple lemma provides the tools for a precise implementation of this idea. 
\begin{lemma}
\label{Lem:MOSM/PSM prop}
{\upshape (i)} If $r_{k+1}^{(i-1)}\geq r_{0}^{(i-1)}\geq r_k^{(i-1)}$, there exists a $\ket{\vec{x}^{(i)}}$ and a T1 measurement by party $k$ transforming $\ket{\vec{x}^{(i-1)}}\to\ket{\vec{x}^{(i)}}$ such that $r_{k+1}^{(i)}\geq r^{(i)}_{k}=r^{(i)}_{0}$ and $s_ip_i=\frac{r_k^{(i-1)}}{r^{(i-1)}_{0}}$.  {\upshape (ii)} If $r_{k+1}^{(i-1)}\geq r^{(i-1)}_{k}\geq r^{(i-1)}_{k-1}$, there exists a $\ket{\vec{x}^{(i)}}$ and a T2 measurement by party $k$ transforming $\ket{\vec{x}^{(i-1)}}\to\ket{\vec{x}^{(i)}}$ such that $r_{k+1}^{(i)}\geq r^{(i)}_{k}=r^{(i)}_{k-1}$.
\end{lemma}
\begin{proof}
(i)  In any T1 measurement we have $r_k^{(i)}=\frac{r_k^{(i-1)}}{p_i}$ and $r^{(i)}_{0}=s_ir^{(i-1)}_{0}\leq s_ir^{(i-1)}_{k+1}=r^{(i)}_{k+1}$.  Setting these equal gives $s_ip_i=\frac{r_k^{(i-1)}}{r^{(i-1)}_{0}}\leq 1$ from which any choice of $s_i$ and $p_i$ satisfying this provides a realizable protocol.  (ii)
For a T2 measurement, $r_{k-1}^{(i)}=\frac{r_{k-1}^{(i-1)}}{p_i}\leq\frac{r_{k+1}^{(i-1)}}{p_i}=r^{(i)}_{k+1}$ and $y_{k}r^{(i)}_{k}=1-\frac{1}{p_i}(1-y_{k}r^{(i-1)}_{k})$.  Equality is achieved with the choice $p_i=1-y^{k}(r^{(i-1)}_{k}-r^{(i-1)}_{k-1})$.
\end{proof}

We now state the result of the protocol as a theorem and then give its proof by constructing the transformation procedure.
\begin{theorem}
\label{Thm:MainLowerBound}
Let $\ket{\vec{x}}$ and $\ket{\vec{y}}$ be two W-class states with $r_1\leq r_2\leq\cdots\leq r_N$ where $r_k=\frac{x_k}{y_k}$ and $y_k\not=0$ for $k=0,\cdots,N$.  If $r_1\geq r_0$, then $\ket{\vec{x}}$ can be converted to $\ket{\vec{y}}$ with probability $r_1$.  Otherwise, let $h$ be the largest integer such that $r_{0}> r_{h}$.  Then $\ket{\vec{x}}$ can be converted to $\ket{\vec{y}}$ with probability \[r_h\left(\frac{r_{h-1}}{r_0}\right)\dots\left(\frac{r_1}{r_0}\right).\]
\end{theorem}
\noindent\textbf{Protocol -}

The protocol can be divided into two parts: in the first only T1 measurements are performed, and in the second only T2s.  If $r_1>r_0$, proceed to the second part and let $h=1$, $p_1=1$, and $r_i^{(1)}=r_i$ for all parties $i$.  In round one, party $h$ performs a T1 such that $r_h^{(1)}=r_0^{(1)}$ with $s_1p_1=\frac{r_h}{r_0}$.  In round two, party $h-1$ performs an T1 such that $r_{h-1}^{(2)}=r_0^{(2)}=r_h^{(2)}$ with $s_2p_2=\frac{r_{h-1}^{(1)}}{r^1_{(0)}}=\frac{r_{h-1}}{r_0}$.  This process is continued for $h$ rounds.  The end result is $r^{(h)}_h=s_{h}r_h^{(h-1)}=s_hs_{h-1}s_{h-2}\cdots s_2\frac{r_h}{p_1}$, and $r^{(h)}_n\geq\cdots\geq r^{(h)}_{h+1}\geq r^{(h)}_h=r^{(h)}_{h-1}=\cdots\geq r_0^{(h)}$, where the last inequality is always tight when $r_0\geq r_1$.  The next part of the protocol now begins with only T2 measurements performed.  In round $h+1$ party $h+1$ performs a T2 with probability $p_{h+1}$ such that $r_{h+1}^{(h+1)}=r^{(h+1)}_h=\cdots\geq r_0^{(h+1)}$ with $r^{(h+1)}_{h}=\frac{1}{p_{h+1}}r^{(h)}_h$.  Next, party $h+2$ performs a T2 with probability $p_{h+2}$ such that $r^{(h+2)}_{h+2}=r^{(h+2)}_{h+1}=r^{{h+2}}_h=\frac{1}{p_{h+2}p_{h+1}}r^{(h)}_h$.  This process is continued until right before some round $l$ in which $r_{l-1}^{(l-1)}\leq 1$ and $r^{(l-1)}_l\geq 1$.  Note such a round will exist because in every round $j$ satisfying $r_0^{(j)}\leq r_N^{(j)}$, there is always some component $i\not=0$ with $r_i^{(j)}\geq 1$.  Returning to the protocol, in round $l$, party $l$ applies a T2 measurement with $p_{l}=r_{l-1}^{(l-1)}$.  As a result, $r^{(l)}_i\geq 1$ for all $i\geq 1$ since $r_i^{(l)}=1$ for $1\leq i< l$, $r_l^{(l)}\geq 1$ by \eqref{Eq:coords change MOSM/PSM}, and $r_j^{(l)}\geq r_l^{(l)}$ for $j\geq l$.  The total probability is 
\begin{align*}
&p_1p_2...p_{l-1}p_{l}=\frac{p_1p_2...p_{l-1}}{p_{h+1}p_{h+2}...p_{l-1}}r^{(h)}_h=p_1\cdots p_h r^{(h)}_h\\
&=p_1\cdots p_h s_hs_{h-1}s_{h-2}\cdots\frac{r_h}{p_1}=r_h\left(\frac{r_{h-1}}{r_0}\right)\dots\left(\frac{r_1}{r_0}\right).
\end{align*}

In order for this protocol to be suitable for any W-class transformation, we must consider the cases when $y_k=0$.  If $y_k=0$ for $k\geq 1$, then the $k^{th}$ party simply first disentangles itself with probability one from the rest of the system and the above protocol is performed on the $N-1$ party state $\ket{\vec{x}'}$ where $x'_0=x_0+x_k$ and $x'_j=x_j$ with $j\geq 1$.  If $y_k=0$ for $k=0$, then party $i$ specified by $x_i=\max_{j\geq 1}\{x_j\}$ performs the filter $M=\sqrt{\lambda}\left(\begin{smallmatrix}1&-\sqrt{\frac{x_0}{x_i}}\\0&1\end{smallmatrix}\right)$ with success probability $\lambda(1-x_0)$ where $\lambda=\frac{2x_i}{x_0+2x_i+\sqrt{x_0^2+4x_ix_0}}$.  This changes the coordinates as $x_i\to\frac{x_i}{1-x_0}$, and hence the constructed protocol can be implemented with an overall success probability of $\lambda r_1$.  

Our protocol is most general in that it and the derived success probability apply to all W-class transformations, even those whose target state is not $N$-partite entangled.  As an example, we compute the probability for an arbitrary W-state distillation.
\begin{corollary}
Let $\ket{\vec{x}}$ be an $N$-party W-class state.  Then
\[P_{max}(\ket{\vec{x}}\to\ket{W_N})\geq \frac{2x_Nx_1N}{x_0+2x_N+\sqrt{x_0^2+4x_Nx_0}}.\]
\end{corollary}

Observe that in Theorem \ref{Thm:MainLowerBound}, the lower bound becomes $r_1$ whenever $r_2\geq r_0$.  Combined with the results of the previous section, we see that $P_{max}(\ket{\vec{x}}\to\ket{\vec{y}})=r_1$ if and only if $r_2\geq r_0$.   

\section{General Features of Symmetric Transformations}

The symmetric W-class states constitute a one parameter family of states which make them easier to analyze.  Any such state can be represented as $\ket{s}=\sqrt{1-s}\ket{\vec{0}}+\sqrt{\frac{s}{N}}\sum_{i=1}^N\ket{\vec{i}}$.  Note that the state $\ket{W_N}$ corresponds to $s=1$.  The optimal probability of converting $\ket{s}\to\ket{t}$, at least by a one-shot measurement, can be numerically computed by brute force using Lagrange multipliers.  However, in this section we are less concerned with analytic expressions for optimal conversion probabilities and more with general properties of transforming symmetric states.

In particular, it was recently shown that one multi-qubit symmetric state can be reversibly converted into another if and only if the transformation is feasible by a protocol in which each party performs the same one-shot measurement.  However, one question that was not investigated is whether the optimal one-shot success probability can always obtained by a symmetric filter.  Here, we answer this question by examining the transformation of an arbitrary tripartite symmetric W-class state to the target state $\ket{W_3}$.  An optimal symmetric one-shot measurement with success probability $q$ can be expressed as $(A\otimes A\otimes A)\ket{s}=\sqrt{q}\ket{W_3}$ with the operator $A$ satisfying $A^\dagger A\leq I$ and $\det(I-A^\dagger A)=0$.  For comparison, we will consider the same transformation when only a single party acts non-trivially: $(A\otimes I\otimes I)\ket{s}=\sqrt{p}\ket{W_3}$.  Note that any conversion among symmetric W-class states can always be achieved with a non-zero probability by the action of just a single party up to a local basis change.

Studying the difference in optimal conversion rates between a multiparty symmetric filter and a single party filter is of interest because it sheds light on two competing intuitions.  On the one hand, when more parties act, the ``work of conversion'' can be distributed, and in light of the overall symmetry, it seems reasonable to expect that it's best for this work to be shared equally in the form of identical filters.  On the other hand, if only one party performs a measurement, there are fewer possibilities for failure.  Below we show that neither of these intuitions are true in general.

In both cases without loss of generality we can take $A$ to have the form $A=\left(\begin{smallmatrix}a&b\\0&c\end{smallmatrix}\right)$.  Then the comparative optimization problems become
\begin{align}
&\max \;p &&\max \;q\notag\\
&\text{subject to:}&&\text{subject to:}\notag\\
&a\sqrt{1-s}+b\sqrt{\frac{s}{3}}=0, &&a^3\sqrt{1-s}+3ba^2\sqrt{\frac{s}{3}}=0,\notag\\
&c\sqrt{s}=\sqrt{p}, &&a^2c\sqrt{s}=\sqrt{q},\notag\\
&a=c,&&
\end{align}
with both satisfying the common constraint that $b^2=(1-a^2)(1-c^2)$.  The respective solutions as functions of $s$ are 
\begin{align*}
p_{max}(s)&=\frac{1}{2}(3-s-\sqrt{3(1-s)(3+s)}),\\ 
q_{max}(s)&=\frac{(3+9s-\beta(s))^2(-3+3s+\beta(s))}{48(1+2s)(1-s+\beta(s))},
\end{align*}
where $\beta(s)=\sqrt{3(1-s)(3+5s)}$.

The difference $p_{max}(s)-q_{max}(s)$ is plotted in Fig. \ref{plot} as $s$ varies between zero and one.  From it, we see that in general, the optimal strategy for transforming two symmetric states involves neither a symmetric measurement nor the action of just a single party.  For the particular class of transformations we consider here, the two strategies have the same maximum efficiency only when $s=\frac{3}{61}(3+8\sqrt{3})$.  However, the difference in optimal probabilities is never greater than 1.4\%.  It should also be noted that neither of these schemes may be the overal optimal protocol.  While further numerical analysis could provide an answer to this question, we do not persue it here as we consider the reported result of greater interest.
\begin{figure}[h]
\includegraphics[scale=0.6]{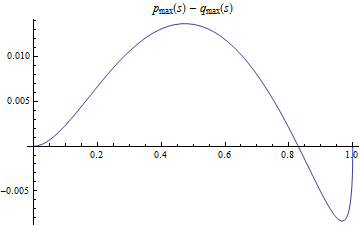}
\caption{\label{plot} The difference in maximum transformation probabilities when only one party measures ($p_{max}(s)$) versus an identical filter by all parties ($q_{max}(s)$).}
\end{figure}
\section{Conclusion}
In this article, we have investigated the LOCC convertibility of $N$-party W-class states.  For a large family of transformations, we have proven their optimal conversion rate to achieve the upper bound of $\min_i\{\frac{x_i}{y_i}\}$.  The question of transforming symmetric states was considered in the context of W-class states.  Despite the necessity of SLOCC equivalent states to be related by a symmetric measurement, we have found that this symmetry cannot be extended to the measurement achieving optimality.  A future direction of research might involve considering the W-class transformations when $r_0>r_2$.  However, preliminary numerical work on this problem has revealed the computation to be quite unyielding.  It would also be interesting to know when the lower bound of Theorem \ref{Thm:MainLowerBound} is optimal.  

{\bf Acknowledgments}
\\
We thank Benjamin Fortescue for helpful discussions in the development of this work as well as support from the funding agencies CIFAR, CRC, NSERC, and QuantumWorks.

\bibliography{EricQuantumBib}

\end{document}